\documentclass[12pt]{article}

\usepackage{graphicx}
\usepackage{amssymb}

\newcommand{\e}{\mathrm{e}}
\newcommand{\D}{\mathrm{d}}
\newcommand{\C}{\mathbb{C}}
\newcommand{\N}{\mathbb{N}}
\newcommand{\R}{\mathbb{R}}
\newcommand{\Da}{\Delta_\alpha}
\newcommand{\Oo}{\mathcal{O}}
\newcommand{\Us}{\upharpoonright_{\Sigma}}
\newcommand{\Ug}{\upharpoonright_{\Gamma}}
\newcommand{\ULi}{\upharpoonright_{\Lambda_{i}}}
\newcommand{\ULj}{\upharpoonright_{\Lambda_{j}}}
\newcommand{\s}{\alpha ,\Sigma}

\newcommand{\bIe}{I^c _\epsilon}
\newcommand{\Ie}{I_\epsilon}
\newcommand{\Qle}{Q_{\lambda }^\epsilon}

\newcommand{\Ple}{P_{\lambda }^\epsilon}
\newcommand{\Qlec}{Q_{\lambda }^{\epsilon c}}
\newcommand{\Hag}{H_{\alpha \Gamma}}

\newcommand{\Hm}[1]{\leavevmode{\marginpar{\tiny%
$\hbox to 0mm{\hspace*{-0.5mm}$\leftarrow$\hss}%
\vcenter{\vrule depth 0.1mm height 0.1mm width \the\marginparwidth}%
\hbox to
0mm{\hss$\rightarrow$\hspace*{-0.5mm}}$\\\relax\raggedright #1}}}

\newtheorem{claim}{Claim}[section]
\newtheorem{theorem}[claim]{Theorem}
\newtheorem{lemma}[claim]{Lemma}

\newtheorem{remark}[claim]{Remark}

\newtheorem{proposition}[claim]{Proposition}
\newenvironment{proof}[1][Proof]{\textsl{#1.} }{\ \rule{0.4em}{0.7em}}

\begin{document}

\title{Straight quantum layer  with impurities inducing resonances  }
\date{}
\author{ Sylwia Kondej}
\maketitle

\begin{center}
{Institute of Physics, University of Zielona G\'{o}ra, ul.
Szafrana 4a,  65246 Zielona G\'{o}ra, Poland}

\end{center}

e-mail: {\it skondej@if.uz.zgora.pl},

\begin{abstract}
We consider a straight three dimensional quantum layer
with singular potential supported on a straight wire which is
localized perpendicularly to the walls and  connects them. We prove
that the infinite number of  embedded eigenvalues appears
in this system. Furthermore, we show that after introducing
a small surface impurity to the layer, the embedded eigenvalues turn to
the second sheet resolvent poles which state
resonances. We discuss the asymptotics of the imaginary component of the resolvent
pole with respect to the surface area.
\end{abstract}



\medskip

{\bf Keywords:} Singular perturbations,
embedded eigenvalues, resonances.

\section{Introduction}\label{introduction}

The paper belongs to the line of  research often called Schr\"odinger operators with delta potentials
\footnote{In the following we will equivalently use the notations delta interaction and delta
potential.}.
The analysis of these type of potentials is motivated by mesoscopic physics systems with the semiconductor
structures  designed in such a way that they  can be mathematically modelled by the Dirac delta  supported on the sets
of lower dimensions. The support of delta potential imitates the  geometry of the semicondutor material, for example,
it can take a form of one dimensional sets (wires) or surfaces with specific geometrical properties. A particle is confined
in the semiconductor  structure however the model admits a possibility of tunneling.
Therefore these types of systems are called in literature \emph{leaky quantum graphs} or \emph{wires}. One of the most
appealing problem in this area is the question  how the geometry of a wire affects the spectrum;
cf.~\cite{AGHH}~and~\cite[Chap.~10]{EK-book}. The aim of the present paper is to discuss
how the surface perturbation leads to resonances.

We consider a non relativistic three dimensional  model of  quantum  particle
confined between two infinite unpenetrable parallel walls which  form
a straight quantum layer defined by $\Omega :=\{(\underline{x}, x_3)\in  \R^2 \times [0, \pi ]\}$.
 In the case of absence
of any additional potential  the Hamiltonian of such system is given  by the negative Laplacian
 $-\Delta\,:\, \mathrm{D}(\Delta )
\,:\, \to L^2 (\Omega )$ with the domain $\mathrm{D}(\Delta )=W^{2,2}(\Omega )\cap W^{1,2}_0(\Omega )$, i.e.~with  the
Dirichlet boundary conditions on $\partial \Omega $.
The spectrum of   $-\Delta $ is determined by $\sigma_{\mathrm{ess}}(-\Delta )=[1, \infty )$ however it is useful to keep
in mind that the energies in $x_3$-direction  are quantized and given by $\{k^2\}_{k=1}^\infty $.

At the first stage we introduce a straight wire $I$ which connects  the walls $\partial \Omega $
being at the same time perpendicular to them. We assume
the presence of interaction localized on $I$ and characterized by the coupling constant $\alpha \in \R$. The symbolic Hamiltonian of
such system can be formally written
\begin{equation}\label{eq-formal}
-\Delta + \delta_{\alpha , I}\,,
\end{equation}
where $\delta_{\alpha ,I}$ represents delta potential supported on $I$.  Since the interaction support in this model has the co-dimension
larger then one it is called \emph{strongly singular potential}.
 The proper mathematical definition of Hamiltonian
can be  formulated  in the terms of boundary conditions. More precisely, we define $H_\alpha $ as  a self adjoint extension
of $-\Delta \left|_{C^{\infty}_0 (\Omega \setminus I)}\right.$ determined by means of the appropriate boundary conditions which functions
from the domain $\mathrm{D}(H_\alpha )$
 satisfy on $I$.
 The coupling constant $\alpha $ is involved in the mentioned boundary
conditions, however, it is worth to say at this point that
 $\alpha $ does not contribute  additively to the structure of Hamiltonian.
\\ To describe spectral  properties of $H_\alpha $ we can rely on radial symmetry of the system
and consider two dimensional system with point interaction governed by the Hamiltonian $H^{(1)}_\alpha $.
The spectrum of $H^{(1)}_\alpha $ consists of positive half line and one discrete negative eigenvalue
$$
\xi_\alpha = -4\e^{2 (-2\pi \alpha +\psi(1))}\,,
$$
cf.~\cite{AGHH}, where $-\psi(1)= 0,577...$ determines the Euler--Mascheroni constant.
This reflexes the structure of spectrum of $H_\alpha $, namely,
for each $l\in \N$ the number
$$
\epsilon_l = \xi_\alpha +l^2\,,
$$
gives rise to an eigenvalue of $H_\alpha $. Note that infinite number  of $\epsilon_l$ lives above the threshold of the
essential spectrum and, consequently, the Hamiltonian $H_\alpha $ admits  the infinite number of \emph{embedded eigenvalues}.
 \\ In the second stage we introduce to the layer
 an attractive interaction supported on a finite $C^2$ surface $\Sigma \subset \Omega $ separated
 from a wire $I$ by some distance, cf.~Fig.~1.  Suppose that  $\beta \neq 0$ is a real number. The Hamiltonian $H_{\alpha , \beta }$ which governs this system can
 be symbolically written  as
 $$
 -\Delta +\delta_{\alpha, I} -\beta \delta_{\Sigma}\,, 
 $$
 where $\delta_{\Sigma}$ stands for the  Dirac delta supported on $\Sigma $; this term represents \emph{weakly
 singular potential}
  Again, a proper mathematical definition of
 $H_{\alpha, \beta }$ can be formulated as a self adjoint extension of $$H_{\alpha }\left|_{\{f\in \mathrm{D}(H_\alpha )
\cap C^\infty (\Omega \setminus I) \,:\, f =0 \,\,\, \mathrm{on}\,\,\,
\Sigma \}}\right. \,.$$
 This extension is defined by means of the
   appropriate boundary conditions on $\Sigma$ discussed in Section~\ref{sec-surface}.

\begin{figure}
\includegraphics[width=.40\textwidth]{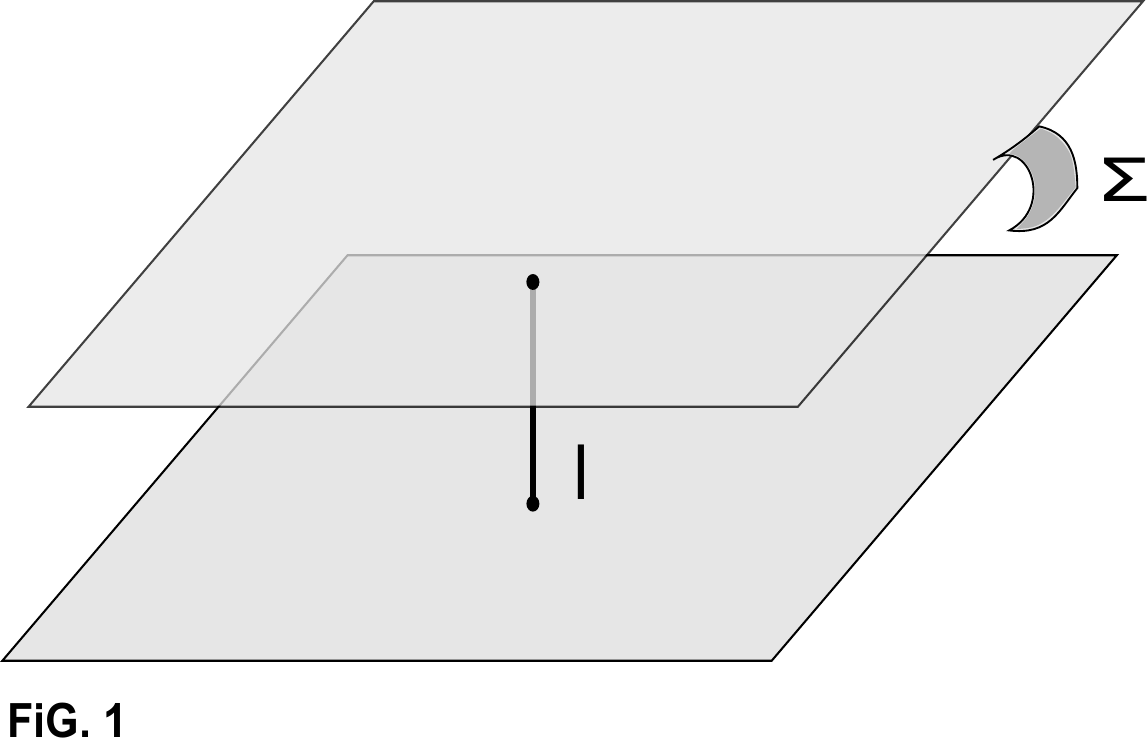}
\end{figure}

 The aim of this paper is to analyse \emph{how the presence of surface interaction supported on
 $\Sigma $ affects the embedded eigenvalues}. The existence of embedded eigenvalues is a direct consequence
 of the symmetry. By introducing additional interaction on $\Sigma$ we break this symmetry, however, if the perturbation
 is small then we may expect that the system preserves a ''spectral memory`` on original eigenvalues.
 In Section~\ref{sec-reson} we show, for example,  that
  if the area $|\Sigma| \to 0  $  then \emph{the embedded eigenvalues $\epsilon_l$
 turn to complex poles of the resolvent
  of $H_{\alpha, \beta }$.} These poles are given by  $z_l = \epsilon_l +o(|\Sigma | ) $ with  $\Im z_l <0$; the latter confirms that $z_l $ is localized on the second sheet continuation.
  We derive the explicit formula for the lowest order of the imaginary component of $z_l$ and show that it admits  the  following asymptotics
  $$\Im z_l = \mathcal{O} (|\Sigma|^2)\,.$$ The poles of resolvent state resonances in the system governed by $H_{\alpha , \beta }$
  and $\Im z_l $ is related to the width of the resonance given by $-2\Im z_l $.
  \\ \\
Finally, let us mention that various  types of resonators in waveguides and layers  have been already analyzed.
For example, in \cite{KS} the authors study resonances induced
by the twisting of waveguide  which is responsible for breaking symmetry.
The planar waveguide with narrows playing the role of resonators
has been studied in \cite{BKNPS}. On the other hand, the straight Dirichlet or Neuman waveguides with windows or barriers
inducing resonances have been analyzed in \cite{BEG, Popov2000, Popov2003}.
Furthermore, resonances in the curved  waveguides with finite branches
have been  described in \cite{DNG}. It is also worth to mention that quantum waveguides with electric and magnetic fields
have been   considered, ~cf.~\cite{BPS, BG}.

On the other hand, the various types of resonators induced by delta potential
in two or three dimensional systems
have been analyzed. Let us mention the results of \cite{EK3, Kondej2012, KK013}
which describe resonances  in the terms of breaking symmetry parameters or by means of tunnelling effect.

In  \cite{KondejLeonski2014} the authors consider a straight two dimensional waveguide with a semitransparent perpendicular
barrier modeled by delta  potential. It was shown that after changing slightly the slope of barrier the embedded eigenvalues turn to resonances;
the widths of these  resonances can be expressed in the terms of the barrier slope.
The present paper is, in a sense, an extension of \cite{KondejLeonski2014}. However,  the strongly singular character of the delta
interaction supported on $I$ causes that even an infinitesimal  change of the slope of $I$, can not be understood as a small perturbation.
Therefore  the resolvent poles are not interesting from the physical point of view
since they rapidly escape far away
from the real line.
In the present model the role of
small perturbation plays  delta potential on $\Sigma$ which
 leads to resonances.

Finally, it is worth to mention that the spectral properties of quantum waveguides and layers with delta
interaction have been studied, for example, in  \cite{EKrecirik99, EN}. The results of \cite{EKrecirik99}
concern weakly singular potentials and in \cite{EN} the authors consider  strongly singular
interaction. In the present paper we combine both  types of delta interaction and analyze how they affect
each other.


  \bigskip

\noindent General notations:
\\ $\bullet$ $\C$ stands for the complex plane and $\C_\pm $ for the upper, respectively,  lower half-plane.

\medskip

\noindent $\bullet$ $\|\cdot \|$, $(\cdot ,\cdot )$ denote the norm and the scalar product
in $L^2 (\Omega )$ and $(\cdot, \cdot )_{\Sigma}$ defines the scalar product in  $L^2 (\Sigma )$.

\medskip
\noindent $\bullet$ Suppose that $A$ stand
for a self adjoint operator. We standardly denote by $\sigma _\mathrm{ess}(A)$,
 $\sigma _\mathrm{p}(A)$ and $\rho (A)$, respectively, the essential spectrum, the point spectrum and the resolvent set of $A$.

\medskip

\noindent  $\bullet$ The notation $C$ stands for a constant which value can vary from line
to line.

\section{Parallel walls connected by wire inducing  embedded eigenvalues}

\subsection{Free particle in layer.}
Let $\Omega \subset  \R^3$ stand for  a layer defined
 by $\Omega := \{ x= ( \underline{x}, x_3)\,:\,\underline{x}
 \in \R^2  \,,x_3 \in [0, \pi ]\}$ and in the following we assume convention
 $\underline{x} = (x_1, x_2)
 \in \R^2 $. \\
 The ''free`` Hamiltonian is determined by  $$H= -\Delta \,:\,
 \mathrm{D}(H)= W^ {2,2} (\Omega )\cap W_0 ^{1,2}(\Omega )\to L^2 (\Omega )$$ and it admits the following decomposition
   \begin{equation}\label{eq-Hamilton}
     H=-\Delta ^{(2)} \otimes I +I \otimes -\Delta^{(1)}\,\quad \mathrm{on}\quad  L^2 (\R^2) \otimes L^2 (0,\pi )\,,
   \end{equation}
 where $ \Delta^{(2)} \,:\, \mathrm{D}(\Delta^{(2)}) = W^{2,2}(\R^2)\to L^2  (\R^2)$ stands for the
 two-dimensional Laplacian and $\Delta^{(1)} \,:\, \mathrm{D}(\Delta^{(1)}) = W^{2,2}(0,\pi)\cap W_0^{1,2}(0,\pi)\to
 L^2 (0,\pi)$ determines one-dimensional Laplacian with the Dirichlet boundary conditions.

 To define the resolvent  of $H$  it is useful to note that
 the sequence $\{\chi_n\}_{n=1}^\infty$ given by   $$\chi_n (x_3) : = \sqrt{\frac{2}{\pi}}\sin ( n x_3)\,,\quad  n\in \N$$
 forms an
 orthonormal basis in $L^2 (0, \pi)$. Suppose that $z\in \C\setminus [1,\infty )$. Then
$R(z):= (-\Delta -z)^{-1}$ defines an integral operator with the kernel
 \begin{equation}\label{eq-integralH}
   \mathcal{G}(z; \underline{x}, \underline{ x}',x_3,x'_3):= \frac{1}{2\pi} \sum_{n=1}^\infty K_0 (\kappa_n (z)
   |\underline{x}- \underline{x}'|) \chi_n (x_3)\chi_n (x_3 ' )\,,
 \end{equation}
 where $K_0 (\cdot)$ denotes the Macdonald function, cf.~\cite{AS}, and
 \begin{equation}\label{eq-defk}
\kappa _n (z):= -i\sqrt{z-n^2}\,, \quad \Im \sqrt{z-n^2} >0 \,.
\end{equation}
In the following we will also use the abbreviation $\mathcal{G} (z)$ for (\ref{eq-integralH}). The threshold of spectrum of $H$ is determined by the lowest discrete transversal energy, i.e. $1$.
 Moreover, it is purely absolutely continuous and consequently, it takes the form
 $$
 \sigma (H) = [1,\infty )\,.
 $$
 \subsection{Layer with perpendicular wire: embedded eigenvalues phenomena.}
We introduce a  wire defined by the straight segment of width $\pi$ and perpendicular to walls.
 The presence of the wire will be modelled by delta interaction supported on $I\subset \Omega $, where
  $I:= (0,0)\times [0,\pi]$.
 \\
 In view of the radial symmetry the operator with delta interaction on $I$ admits a natural
 decomposition on
 $L^2 (\Omega)= L^2 (\R^2)\otimes L^2 (0,\pi )$ and acts in the subspace $L^2 (\R^2)$ as the Schr\"odinger operator
 with one point interaction. Therefore, the  delta potential can be determined by appropriate boundary conditions, cf.~\cite[Chap.~1.5]{AGHH},
 which can be implemented, in each sector of the transversal
energy, separately. For this aim we decompose a function  $\psi\in L^2 (\Omega )$ onto
$\psi (x)=\sum_{n=1}^\infty \psi_n (\underline{x} )\chi_n (x_3)$, where
 $\psi_n (\underline{x} ):= \int_{0}^\pi \psi (\underline{x}, x_3 ) \chi_n (x_3)\mathrm{d}x_3$.

 \begin{description}
   \item[$D_1$)]  We say that a function $\psi $ belongs to the set $D'\subset
 W^{2,2}_{\mathrm{loc}} (\Omega \setminus I) \cap L^2 (\Omega )$ if $\Delta \psi \in L^2 (\Omega )$,  $\psi|_{\partial \Omega }=0$ and the following limits
 $$
 \Xi_n (\psi ):= - \lim_{|\underline{x}|\to 0  }\frac{1}{\ln |\underline{x}|} \psi_n (\underline{x})\,, \qquad  \Omega_n (\psi):=
 \lim_{|\underline{x}|\to 0  } \left( \psi_n (|\underline{x} |) - \Xi_n (\psi )\ln |\underline{x} |\right)$$
 are finite.
   \item[$D_2$)] For  $\alpha \in \R$, we  define the set
 \begin{equation}\label{eq-bcalpha}
 \mathrm{D} (H_\alpha ):= \{ \psi \in D' \,:\, 2\pi \alpha \Xi_n (\psi ) = \Omega_n (\psi)\,\,\,\mathrm{for}\,\,\,\mathrm{any}\,\,\,n\in \N \}\,
 \end{equation}
 and the operator $H_\alpha \,:\,  \mathrm{D} (H_\alpha ) \to L^2 (\Omega )$ which acts
 $$ H_\alpha \psi (x)= -\Delta \psi (x)\,,\quad \mathrm{for }\quad x\in \Omega \setminus I\,.$$
 \end{description}
The resulting operator $H_\alpha \,:\,D(H_\alpha )\to L^2 (\Omega )$  coincides
   \begin{equation}\label{eq-Hamilton}
    -\Delta_\alpha  ^{(2)} \otimes I +I \otimes -\Delta^{(1)}\,\quad \mathrm{on}\quad  L^2 (\R^2) \otimes L^2 (0,\pi )\,,
   \end{equation}
 where $ \Delta_\alpha ^{(2)} \,:\, \mathrm{D}(\Delta_\alpha ^{(2)}) \to L^2  (\R^2)$ stands for the
 two-dimensional Laplacian with point interaction, cf.~\cite[Chap.~1.5]{AGHH} with the domain $\mathrm{D}(\Delta_\alpha ^{(2)}) $. Consequently, $H_\alpha $
  is self adjoint and its spectral properties will be discussed in the next section.

\subsection{Resolvent of $H_\alpha $.} 

Suppose that $z\in \C_+$.
We  use the standard notation $R_\alpha (z)$ for the resolvent operator, i.e. $R_\alpha (z):= (H_\alpha - z)^{-1}$.
To figure out the explicit resolvent formula we introduce
$$
\omega_n (z; x):= \frac{1}{2\pi } K_0 (\kappa_n (z)| \underline{x} | ) \chi_n (x_3)\,, \quad n\in \N\,;
$$
in the following we will use also abbreviation $\omega_n (z)= \omega_n (z; \cdot )$.

The following theorem states the desired result.
\begin{theorem}\label{th-resl1} The essential spectrum of $H_\alpha $ is given by
\begin{equation}\label{eq-ess}
\sigma_{\mathrm{ess}} (H_\alpha )=[1, \infty )\,.
\end{equation}
Furthermore, let \footnote{Analogously as in the previous discussion
 we assume $\Im \sqrt{z-n^2} >0$. The logarithmic function $z\mapsto \ln z $ is defined  in the cut
 plane $-\pi <\arg z <\pi$
 and admits continuation to entire logarithmic Riemann surface.}
\begin{equation}\label{eq-defGamma}
\Gamma_n (z):= \frac{1}{2\pi } \left( 2\pi \alpha +s_n (z)\right)\,,\quad
 \mathrm{where }\quad s_n (z):=-\psi(1) +\ln \frac{\sqrt{z-n^2}}{2i}\,.
\end{equation}
Suppose that $z\in \C\setminus [1,\infty )$ and $\Gamma _n (z) \neq 0$. Then $z\in \rho (H_\alpha )$ and operator $R_\alpha (z)$ admits the Krein-like form:
\begin{equation}\label{eq-resolalpha}
R_\alpha(z)= R(z) +\sum_{n=1}^\infty  \Gamma_n (z)^{-1}  (\omega_n (\bar{z} ), \cdot )\omega_n(z)\,.
\end{equation}
\end{theorem}
\begin{proof}
Our first aim is to show that (\ref{eq-resolalpha}) defines the resolvent of $H_\alpha $.
Operator $H_\alpha $ is defined as the self adjoint extension of $-\Delta |_{C^\infty _0 (\Omega \setminus I)}$.
Suppose that $f\in C^\infty _0 (\Omega \setminus I)$. Then $g:=(-\Delta -z )f\in C^\infty _0 (\Omega \setminus I)$.
Employing the fact that $\omega_n (z)=\mathcal{G}(z)\ast (\delta \chi_n )$,
where $\mathcal{G}(z)$ is the kernel defined
by (\ref{eq-integralH}) and $\delta =\delta (\underline{x})$,
we  conclude  that $(\omega_n (\bar{z} ),  g)= \langle \delta \chi_n, f\rangle_{-1,1} = 0$
 where  $\langle\cdot , \cdot \rangle_{-1,1}$ states the duality
between $W^{-1,2}(\Omega )$ and $W^{1,2}(\Omega )$. This, consequently,
 implies
$R_\alpha (z) (-\Delta -z )f = R (z) (-\Delta -z )f= f$ in view of
(\ref{eq-resolalpha}) which means that $R_\alpha (z)$ defines the resolvent of a self adjoint
extension of $-\Delta |_{C^\infty _0 (\Omega \setminus I)}$. To complete the proof we have to show that
any function $g=R_\alpha (z)f$ satisfies boundary conditions (\ref{eq-bcalpha}).
In fact, $g$ admits the unique decomposition $g=g_1+g_2$, where $g_1:= R(z)f$ and
$g_2= \sum_{n=1}^\infty  \Gamma_n (z)^{-1}  (\omega_n (\bar{z} ), f )\omega_n(z)$.
Therefore, a nontrivial contribution to $\Xi_n (g)$ comes from $g_2$ since $g_1 \in W^{2,2}(\Omega )$.
Employing the asymptotic behaviour
of the Macdonald function, cf.~\cite{AS}
\begin{equation}\label{eq-Kexp1}
K_0 (\rho)=  \ln \frac{1}{\rho} +\psi(1)+\mathcal{O}(\rho)\,,
\end{equation}
 we get $\Xi_n (g)= \frac{1}{2\pi } \Gamma_n  (z)^{-1}(\omega_n (\bar{z}), f)$ and
 $$\Omega_n (g)= (1-  \frac{1}{2\pi } \Gamma_n  (z)^{-1}s_n (z) ) (\omega_n (\bar{z}), f) =
 \alpha \Gamma_n  (z)^{-1} (\omega_n (\bar{z}), f)\,.$$
 Using (\ref{eq-defGamma}) one obtains  (\ref{eq-bcalpha}). This completes the proof of (\ref{eq-resolalpha}).
 The stability of the essential spectrum can be concluded in the analogous way as in \cite[Thm.~3.1]{BEKS}. The key step
 is to show that $R(z)-R_\alpha (z)$ is compact. The statement can be proved relying on
compactness of the trace map $S\,:\, 
W^{2,2}(\Omega )\to L^2 (I ) $ which follows from the boundedness of the trace map, cf.~\cite[Chap.~1,~Thm.~8.3]{LM} and the compactness theorem, cf.~\cite[Chap.~1,~Thm.~16.1]{LM}.
 This implies, in view of boundedness of $R(z)\,:\, L^2 (\Omega )\to  W_0^{1,2}(\Omega )\cap W^{2,2}(\Omega )$, that
 $SR(z)\,:\, L^2 (\Omega )\to L^2(I)$ is compact. Employing the resolvent formula, cf.~\cite{Po}, and the fact
 that the remaining operators contributing to
   $R(z)-R_\alpha (z)$  are bounded we conclude that     $R(z)-R_\alpha (z)$  is  compact.
\end{proof}

\begin{remark} {\rm The spectral analysis developed in this work is mainly   based
on  the resolvent properties. In the following we will use the results of
\cite{BEKS, BEHL, Po, Po2} where strongly as well as weakly singular potentials were considered.
}
\end{remark}

In the following theorem we state the existence of eigenvalues of $H_\alpha $.
\begin{theorem} \label{th-ev}
  Let $\mathcal{A}_\alpha := \{n\in \N
  \,:\, \xi_\alpha +n^2 <1 \}$. Each $\epsilon_n:= \xi_\alpha +n^2$, where $n\in \mathcal{A}_\alpha $ defines the discrete
  eigenvalue of $H_\alpha $ with the corresponding eigenfunction $\omega_n:= \omega_n (\epsilon_n)$. In particular,  this means that for any  $\alpha$
  operator $H_\alpha $ has at least one eigenvalue $\epsilon_1$ below the threshold of the essential spectrum.
  \\ Operator $H_\alpha $ has infinite number of embedded eigenvalues. More precisely, for any $n\in \N \setminus \mathcal{A}_\alpha $
  the number $\epsilon_n:= \xi_\alpha +n^2$ determines the embedded eigenvalue. In particular,
  there exists $\tilde{n} \in \N \setminus \mathcal{A}_\alpha $
  such that $\epsilon_n \in ((n-1)^2, n^2)$ for any $n> \tilde{n}$.
\end{theorem}
\begin{proof}
The proof is based on the Birman-Schwinger argument which, in view of (\ref{eq-resolalpha}), reads
$$z\in \sigma_\mathrm{p} (H_\alpha )\quad \Leftrightarrow  \quad  \exists \, n\in \N\,:\,
\Gamma _n(z)=0 \,,$$ cf.~\cite[Thm.~2.2]{Po2}.  Note that, given $n\in \N$
the function $z\mapsto \Gamma_n (z)$, $z \in \{ \C\,:\, \Im \sqrt{z-n^2}> 0\} $ has the unique
zero at $z=\xi_\alpha +n^2$, i.e.
$$
\Gamma_n (\xi_\alpha +n^2) = 0\,.
$$
Finally, it follows, for example, from \cite[Thm.~3.4]{Po2} that the corresponding eigenfunctions takes the form $\mathcal{G} (\epsilon_l)\ast \chi_n \delta $. This completes the proof.
\end{proof}

\section{Surface impurity } \label{sec-surface}

We define  a finite smooth parameterized surface  $\Sigma \subset \Omega $ being  a graph of the map $ U\ni q=(q_1, q_2) \mapsto x (q) \in \Omega $. 
The surface element can be calculated by means of the standard formula
$\mathrm{d} \Sigma = |\partial_{q_1}x (q)\times  \partial_{q_1}x (q) |\mathrm{d}q$.
Additionally we assume that  $\Sigma \cap I = \emptyset$.
Furthermore, let $n\,:\, \Sigma \to \R^3 $ stand for the unit normal vector (with an arbitrary orientation) 
and $\partial_n  $ denote the normal derivative defined by vector $n$.
Relying on the Sobolev theorem we state that the trace map $W^{1,2}(\Omega ) \ni \psi \mapsto \psi |_{\Sigma } \in  L^2 (\Sigma )$ constitutes
a bounded operator;  we set the notation $(\cdot , \cdot )_{\Sigma }$ for the
scalar product in $L^2 (\Sigma )$.
Given $\beta \in \R\setminus \{0\}$ we define the following boundary conditions: suppose
that  $\psi \in  C (\Omega ) \cap C^1 (\Omega\setminus \Sigma  )  $
satisfies
\begin{equation}\label{eq-bc2}
\partial_n ^+ \psi|_{\Sigma } - \partial_n ^- \psi|_{\Sigma }= -\beta  \psi|_{\Sigma }\,,
\end{equation}
where  the partial derivatives contributing  to the above expression
are defined as the positive, resp., negative limits on $\Sigma $ 
and signs are understood with respect
to direction of $n$.
 \begin{description}
   \item[$D_3$)]  We say that a function $\psi $ belongs to the set $\breve{D}\subset
 W^{2,2}_{\mathrm{loc}} (\Omega \setminus (I\cup \Sigma ))$ if $\Delta \psi \in L^2 (\Omega )$, $\psi|_{\partial \Omega }=0$ and the limiting equations (\ref{eq-bcalpha}) and (\ref{eq-bc2})   are satisfied.
   \item[$D_4$)] Define operator which  for $f\in \breve{D}$ acts as  $-\Delta f (x)$ if $x\in \Omega \setminus (I\cup \Sigma)$
   and let $H_{\alpha, \beta }\,:\, \mathrm{D} (H_{\alpha, \beta }) \to L^2 (\Omega )$ stand for its closure.
\end{description}

To figure out the resolvent of $H_{\alpha , \beta }$ we  define the operator acting from
to $L^{2} (\Omega  )$ to $L^{2}(\Sigma  )$ as $
R_{\alpha, \Sigma }(z)f= (R_{\alpha }(z)f )|_{L^2(\Sigma )}
$. Furthermore, we introduce the operator from $L^2 (\Sigma )$ to $L^2 (\Omega )$
defined by $\mathrm{R}_{\alpha, \Sigma } (z)f = \mathcal{G}_{\alpha } \ast f \delta $, where $\mathcal{G}_{\alpha }$ stands
for kernel of (\ref{eq-resolalpha}). Finally, we define  $\mathrm{R}_{\alpha, \Sigma \Sigma } (z)\,:\, L^2 (\Sigma ) \to L^2 (\Sigma )$
by $\mathrm{R}_{\alpha, \Sigma \Sigma } (z) f= (\mathrm{R}_{\alpha, \Sigma } (z)f)|_{\Sigma }$. In view of
(\ref{eq-resolalpha}) the latter takes  the following form
\begin{equation}\label{eq-resolalphaembed}
\mathrm{R}_{\alpha, \Sigma \Sigma } (z)= \mathrm{R}_{\Sigma \Sigma }(z) +\sum_{n=1}^\infty  \Gamma_n (z)^{-1}  (w_n (\bar{z} ), \cdot )_{\Sigma }w_n(z)\,,
\end{equation}
where $w_n(z):= \omega_n(z)|_{\Sigma }$ and $ \mathrm{R}_{\Sigma \Sigma }(z)\,:\, L^2 (\Sigma )\to L^2 (\Sigma )$
stands for the bilateral embedding of $R(z)$.

Following the strategy developed in \cite{Po} we define the set $Z \subset \rho (H_\alpha )$ such that $z$ belongs to $Z$ if the operators
$$(I-\beta \mathrm{R}_{\alpha , \Sigma \Sigma}(z))^{-1}\,,\quad \mathrm{and} \quad (I-\beta \mathrm{R}_{\alpha , \Sigma \Sigma}(\bar{z}))^{-1}$$ acting from $L^2(\Sigma )$ to $L^2 (\Sigma)$
   exist and are bounded. 
   Our aim is to show that
   \begin{equation}\label{eq-Z}
     Z\neq \emptyset\,.
   \end{equation}
  Therefore we auxiliary define the quadratic below bounded  form
   $$
   \int_{\Omega } |\psi |^2\mathrm{d}x - \beta \int_{\Sigma} \left| \psi |_\Sigma \right|^2\mathrm{d}\Sigma\,,\quad \psi\in W^{1,2}_0 (\Omega )\,.
   $$  Let $\breve{H}_\beta $ stand for the operator associated to the above form in the sense of the first representation theorem, cf.~\cite[Chap.VI]{Kato}. Following the arguments from~\cite{BEKS}
 we conclude  that $I-\beta \mathrm{R}_{\Sigma \Sigma}(z)\,:\,L^2 (\Sigma)\to L^2 (\Sigma)$ defines  the Birman--Schwinger operator for $\breve{H}_\beta $.  Using Thm.~2.2~of~\cite{Po2} one obtains
 $$
 z \in \rho (\breve{H}_\beta )\,\, \Leftrightarrow \,\, 0\in \rho (I-\beta \mathrm{R}_{\Sigma \Sigma}(z))\,.
 $$
 In the following we are interested in negative spectral parameter and thus we
  assume $z=-\lambda $ where $\lambda >0$.
Since the spectrum of $\breve{H}_\beta$ is lower bounded we conclude
\begin{equation}\label{eq-rhoRss}
0\in \rho (I-\beta \mathrm{R}_{\Sigma \Sigma }(-\lambda) )\,,
\end{equation}
for $\lambda$ large enough.

Next step is to find a bound for the second component contributing to (\ref{eq-resolalphaembed}). In fact, it can be majorized by
   $$
   \sum_{n=1}^\infty  \left| \Gamma_n (-\lambda)^{-1} \right| \|w_n (-\lambda )\|_{\Sigma }^2 \leq C
    \sum_{n=1}^\infty   \|w_n (-\lambda )\|_{\Sigma }^2\,,
   $$
   where we applied the uniform bound  $\left| \Gamma_n (-\lambda)^{-1} \right| \leq C$, cf.~(\ref{eq-defGamma}).
   Using the large argument expansion, cf.~\cite{AS},
   \begin{equation}\label{eq-K0}
 K_0 (z)\sim \sqrt{\frac{\pi }{2z}}\e ^{-z}
  \end{equation}
 we get the estimate
$$
|w_n (-\lambda ,x)| \leq C \frac{1}{\lambda^{1/4}}\e^{-r_{\mathrm{min}}(n^2+\lambda )^{1/2}}\quad
\mathrm{for} \,\,\,\lambda \to \infty\,,
$$
where $r_{\mathrm{min}}= \min _{x\in \Sigma } |\underline{x}|$. This implies that  the norm of
the second component of  (\ref{eq-resolalphaembed}) behave as $o(\lambda^{-1})$. Combining this result with
(\ref{eq-rhoRss}) we conclude that $0\in \rho (I-\beta \mathrm{R}_{\alpha, \Sigma\Sigma} (-\lambda ))$ for $\lambda $ sufficiently large which shows that (\ref{eq-Z}) holds.

To realize the strategy of~\cite{Po} we observe that the embedding operator
$\tau ^\ast \,:\, L^2 (\Sigma )\to W^{-1,2}(\Omega )$ acting as  $\tau ^\ast f= f\ast \delta  $ is bounded and, moreover,
\begin{equation}\label{eq-tau}
  \mathrm{Ran}\, \tau ^\ast \cap L^2 (\Omega ) = \{0\} \,.
\end{equation}
Suppose that $z\in Z$. Using  (\ref{eq-tau}) together with Thm.~2.1~of~\cite{Po} we conclude that the expression
 \begin{equation}\label{eq-resolbeta}
  R_{\alpha, \beta }(z)= R_{\alpha } (z)+  \mathrm{R}_{\alpha, \Sigma } (z)(I-
  \beta \mathrm{R}_{\alpha , \Sigma \Sigma}(z))^{-1} R_{\alpha, \Sigma } (z)\,.
\end{equation}
defines the resolvent of self adjoint operator.

 \begin{theorem}  We have
 $$
   R_{\alpha, \beta }(z)= (H_{\alpha , \beta } - z)^{-1}\,.
 $$
 \end{theorem}
\begin{proof}
To show the statement we repeat the strategy applied in the proof of Theorem~\ref{th-resl1}.
Operator $H_{\alpha, \beta }$ is defined as the self adjoint  extension of $-\Delta |_{C^\infty_0 (\Omega \setminus (I \cup \Sigma ))}$
determined by imposing boundary conditions (\ref{eq-bcalpha}) and (\ref{eq-bc2}). The idea is show that
any function from the domain $\mathrm{D} (H_{\alpha, \beta })$ satisfies (\ref{eq-bcalpha}) and (\ref{eq-bc2}). Since the proof can be done
by the mimicking the arguments from the proof of Theorem~\ref{th-resl1} we omit further details.
\end{proof}

Furthermore, repeating the arguments from the proof of Theorem~\ref{th-resl1} we state that
$$
\sigma_{\mathrm{ess}} (H_{\alpha , \beta }) = [1\,,\infty )\,.
$$
\bigskip

\noindent {\bf Notation.}  In the following we will be interested in the spectral
asymptotic for small  $|\Sigma |$. Therefore, we introduce an appropriate scaling with respect to a point $x_0 \in \Sigma$.
Namely,  for small positive  parameter $\delta$ we define $\Sigma_\delta$ as the graph of $U \ni q \mapsto x_\delta (q) \in \Omega $ where
$$
x_\delta (q):= \delta x(q)-\delta x_0+x_0\,.
$$
For example a sphere of radius $R$  originated at $x_0$  turns to  the sphere of radius $\delta R$
after scaling.   Note that equivalence $|\partial_{q_1}x_\delta (q)\times  \partial_{q_1}x_\delta (q) |=\delta^2
|\partial_{q_1}x(q)\times  \partial_{q_1}x (q) |$ implies the scaling  of the surface area $|\Sigma_\delta |= \delta^2 |\Sigma |$.

 \section{Preliminary results for the analysis of poles}
 The Birman-Schwinger argument relates the eigenvalues of $H_{\alpha, \beta }$ and zeros
 of $I-\beta \mathrm{R}_{\alpha , \Sigma \Sigma }(z)$ determined by the condition
 $\ker (I-\beta \mathrm{R}_{\alpha , \Sigma \Sigma }(z))\neq \{0\}$.
 To recover resonances we show that $\mathrm{R}_{\alpha , \Sigma \Sigma } (z)$ has a second sheet continuation
 $\mathrm{R}_{\alpha , \Sigma \Sigma }^{\mathit{II}} (z)$ and the statement
 \begin{equation}\label{eq-resonance}
   \ker (I-\beta \mathrm{R}^{\mathit{II}}_{\alpha , \Sigma \Sigma }(z))\neq \{0\}
 \end{equation}
 holds for certain  $z\in \C_-$.
\subsection{Analytic continuation of $\mathrm{R}_{\alpha , \Sigma \Sigma }(z)$}
We start with the analysis of the first component  of
$\mathrm{R}_{\alpha , \Sigma \Sigma }(z)$ determined  by $\mathrm{R}_{\Sigma\Sigma }(z)$, cf.~(\ref{eq-resolalphaembed}).
Since $\mathrm{R}_{\Sigma \Sigma }(z)$ is defined by means of the embedding of
  kernel $\mathcal{G} (z)$, see~(\ref{eq-integralH}),
 the following lemma will be useful for further discussion.
\begin{lemma} \label{le-contR}
For any  $k\in \N$ the function  $\mathcal{G}(z)$ admits the second sheet continuation
$\mathcal{G}^{\mathit{II}}(z)$ through $J_k:= (k^2, (k+1)^2))$
to an open set $\Pi_k \subset \C_-$ and $\partial  \Pi_k \cap \R =J_k$. Moreover,
$\mathcal{G}^{\mathit{II}}(z)$ takes the form
$$
\mathcal{G}^{\mathit{II}} (z; \underline{x},\underline{x}', x_3, x_3') = \frac{1}{2\pi} \sum_{n=1}^\infty Z_0 (i\sqrt{z-n^2} |\underline{x}-\underline{x}'|)\chi_n (x_3) \chi_n (x'_3)\,,
$$
where
\begin{equation}\label{eq-defZ}
    Z_0 (i\sqrt{z-n^2}  \rho ) =
     \left\{\begin{array}{lr}
        K_0 (-i\sqrt{z-n^2}
                          \rho ), & \mathrm{for } \quad n>k \\
         K_0 (-i\sqrt{z-n^2}
                          \rho )+i\pi I_0 (i\sqrt{z-n^2}
                          \rho ), & \mathrm{for }  \quad n\leq k \,,
        \end{array}\right.
\end{equation}
  and $I_0 (\cdot)$ standardly denotes  the Bessel function.
\end{lemma}
\begin{proof}
  The proof is based  on edge-of-the-wedge theorem, i.e. our aim is to establish the convergence
  $$
  \mathcal{G}(\lambda+i 0) = \mathcal{G}^{\mathit{II}}(\lambda-i 0)\,,
  $$
  for $\lambda \in J_k $.  In fact, it suffices to show that the analogous formula holds for $Z_0$ and $z=\lambda \pm i 0
 $. \\
  Assume first that $n>k$. Then $\sqrt{\lambda -n^2 \pm i0}= \sqrt{\lambda -n^2 }$ since $\Im \sqrt{\lambda -n^2 } >0$.
  Furthermore,  the  function  $K_0 (\cdot)$ is analytic in the upper half-plane, consequently,
  we have $  K_0 (-i\sqrt{\lambda -n^2\pm i0} \rho )
  = K_0 (\sqrt{n^2 -\lambda } \rho )$.\\
  Assume now that $n \leq k$. Then $\sqrt{\lambda -n^2 \pm i0}= \pm \sqrt{\lambda -n^2 }\in \R$ which implies  \begin{equation}\label{eq-Kext} K_0 (-i\sqrt{\lambda -n^2 + i0} \rho )=
  K_0 (-i\sqrt{\lambda -n^2 } \rho )\,.\end{equation}
On the other hand, using the analytic continuation formulae
 $$
  K_0 (z \e^{m\pi i })=K_0 (z) - i\,m\pi I_0 (z)\, \quad \mathrm{and} \quad I_0  (z \e^{m\pi i }) = I_0 (z)\,,
  $$ for   $m\in \N$, we get \begin{eqnarray}
         \nonumber 
           Z_0 (\sqrt{\lambda -n^2 - i0} \rho  ) &=&
   K_0 (i\sqrt{\lambda -n^2 } \rho )+ i\,\pi  I_0 (-i\sqrt{\lambda -n^2 } \rho )
    \\ \nonumber
            &=&  K_0 (-i\sqrt{\lambda -n^2 } \rho )\,.
         \end{eqnarray} This completes the proof.
\end{proof}\\

\begin{figure}
\includegraphics[width=.60\textwidth]{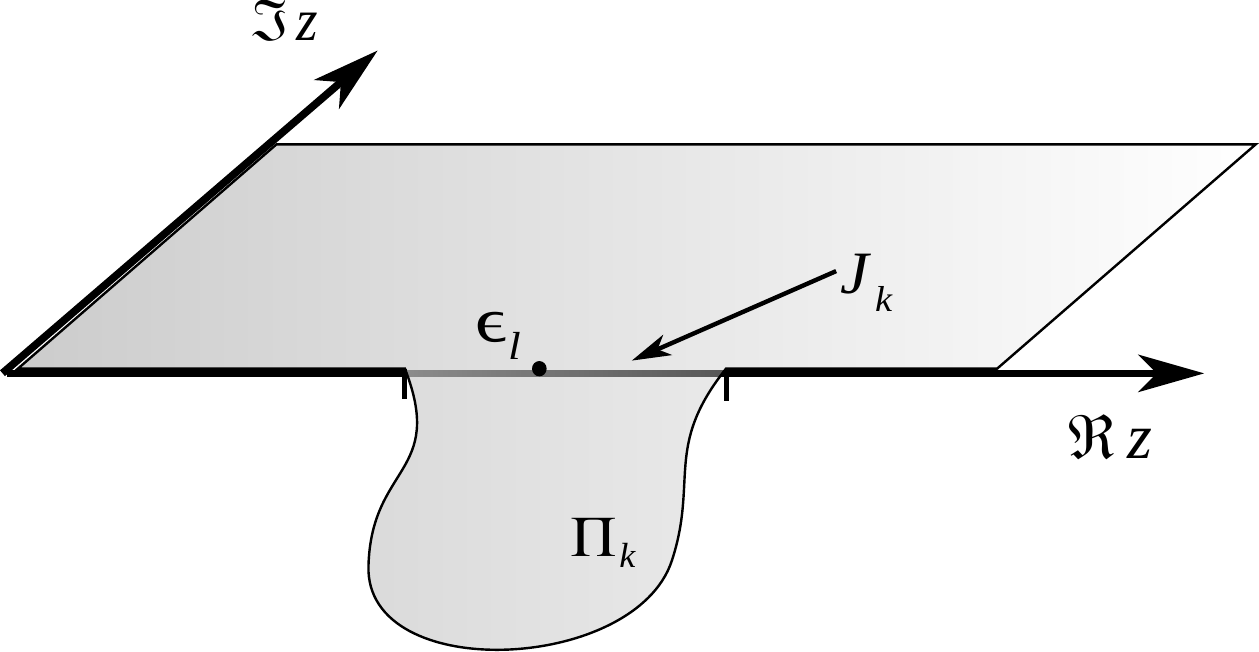}
\end{figure}
The above lemma provides  the second sheet continuation of $R (z)$ as well as  $\mathrm{R}_{\Sigma\Sigma }(z)$;
the latter is defined
as the bilateral embedding of $R^{\mathit{II}} (z)$ to $L^2 (\Sigma )$.
\begin{remark} \rm{
  Note that for each $k\in \N$ the analytic continuation of $\mathcal{G} (\cdot)$ through $J_k $ leads to
  different branches. Therefore, we have to keep in mind that the analytic continuation of $\mathcal{G} (\cdot)$
  is $k$-dependent.}
\end{remark}
In the next lemma we show that the operator $\mathrm{R}^{\mathit{II}}_{\Sigma_\delta  \Sigma_\delta  }(z
)$ is bounded and derive the operator norm  asymptotics if  $\delta \to 0 $.

\begin{lemma} \label{le-boundedR}
  Assume that $k\in \N$  and $\lambda\in J_k$.  Let $z=\lambda -i\varepsilon $, where $\varepsilon $ is a small positive number. Operator $\mathrm{R}^{\mathit{II}}_{\Sigma_\delta \Sigma_\delta} (z) $ is bounded and its norm admits the asymptotics
  \begin{equation}\label{eq-boundedR}
\|
\mathrm{R}^{\mathit{II}}_{\Sigma_\delta \Sigma_\delta} (z) \|  =o(1)\,,
\end{equation}
where the error term is understood with respect to $\delta$.
\end{lemma}
\begin{proof}
To estimate the kernel of $\mathrm{R}^{\mathit{II}}_{\Sigma_\delta \Sigma_\delta} (z)$ we use  (\ref{eq-defZ}), i.e.
\begin{eqnarray} \label{eq-aux0}
\mathcal{G}^{\mathit{II}}_{\Sigma_\delta \Sigma_\delta  }(z; \rho , x_3, x_3') = \frac{1}{2\pi } \left(\sum _{n=1}^\infty K_0 (\kappa_n (z)\rho )\chi_n (x_3)\chi_n (x_3')
\right. \\  \label{eq-aux0a} \left.+\sum _{n\leq k } I_0 (-\kappa_n (z)\rho )\chi_n (x_3)\chi_n (x_3') \right)\,,
\end{eqnarray}
where $$\rho =|\underline{x}-\underline{x}'|\,,\quad \mathrm{and } \quad x\,,x'\in \Sigma_\delta \,.$$
First, we consider (\ref{eq-aux0a}).
The expression  $$\left|I_0 (-\kappa_n (z)\rho )\chi_n (\cdot )\chi_n (\cdot )\right|$$ is bounded.
Therefore the operator defined by the kernel (\ref{eq-aux0a}) is also bounded and the corresponding operator norm in
$L^2 (\Sigma_\delta )$ behaves as $|\Sigma_\delta |^2 =\mathcal{O}(\delta^4 )$. \\
The analysis of the term (\ref{eq-aux0}) is more involving because it consists of infinite number of components.
The asymptotics:
$$
 K_0 (\kappa_n (z)\rho )-K_0 (n \rho ) = \ln \sqrt{1-\frac{z}{n^2}}\left(  1+\mathcal{O}(\rho )\right)\,
$$
implies
\begin{equation}\label{eq-aux1}
   \sum_{n=1}^\infty \left| (K_0 (\kappa_n (z)\rho )-K_0 (n \rho ))\chi_n (x_3)\chi_n (x_3 ') \right|= C+ \mathcal{O}(\rho )\,;
\end{equation}
remind that $\kappa_n (z)$  is defined by (\ref{eq-defk}).
To estimate $\sum_{n=1}^\infty K_0 (n \rho )\chi_n (x_3)\chi_n (x_3 ') $ we borrow the idea
from~\cite{EN} and use~\cite[Chap. 10, II, 5.9.1.4.]{Prudnikov} to get
\begin{eqnarray}
\label{eq-aux2aa} 
 \sum_{n=1}^\infty K_0 (n\rho )\cos (na ) &=& \frac{\pi }{2\sqrt {\rho^2+a^2}}+\frac{1}{2}  \left(\ln \frac{\rho }{4\pi } - \psi(1)\right) \\ \label{eq-aux2}&& +
  \frac{\pi }{2}\sum _{n=1}^\infty \left( \frac{1}{\sqrt{(2n\pi +a)^2+\rho ^2}}-\frac{1}{2n \pi }\right)\\ \label{eq-aux2a}
   && + \frac{\pi }{2}\sum _{n=1}^\infty \left( \frac{1}{\sqrt{(2n\pi -a)^2+\rho^2}}-\frac{1}{2n \pi }\right)\,.
\end{eqnarray}
For $x,x'\in \Sigma_\delta $ the terms (\ref{eq-aux2}) and (\ref{eq-aux2a}) can be majorized by
$C\left(\sum_{n=1}^\infty \frac{1}{n^2}\right)$, i.e. by a uniform constant.
Consequently, using the above estimates  together with the equivalence
$\sin a \sin b =\frac{1}{2}\left( \cos(a-b)-\cos(a+b)\right)$ we get after
straightforward calculations
\begin{equation}\label{eq-13}
  \left| \sum_{n=1}^\infty K_0 (\kappa_n (z)\rho )\chi_n (x_3)\chi_n (x_3 ')   \right| \leq C\left(\frac{1}{|x-x'|}+
  \ln |\underline{x}-\underline{x}'|\right)\,;
\end{equation}
the singular terms in the above estimates come from (\ref{eq-aux2aa}).
Let us analyze the left hand side of (\ref{eq-13}). First, we  consider the component $\mathcal{P}(x,x'):=\frac{1}{|x-x'|}  $ which
gives
$$
\int_{\Sigma_\delta }\mathcal{P} (x,x') \mathrm{d} \Sigma_\delta = (\mathcal{P} \ast \delta_{\Sigma_\delta} )(x) =
\int_{\Sigma_\delta } \frac{1}{|x -x' |} \mathrm{d}\Sigma_\delta \,.
$$
To conclude the desired convergence we employ the concept
of generalized Kato measure. Namely, since the Dirac delta on $\Sigma_\delta $
defines Kato measure we obtain
$$
\sup _{x\in \Sigma_\delta } \int_{\Sigma_\delta } \mathcal{P} (x,x')  \mathrm{d}\Sigma_\delta  = o(1)\,,
$$
where the right hand side asymptotics is understood in the sense of convergence with respect  to $\delta$.
Employing the Schur argument we conclude that the norm  of the integral operator  with the kernel
$\mathcal{P} (x,x') $ acting from $L^2 (\Sigma_\delta )$ to $L^2 (\Sigma_\delta )$ behaves as $o(1)$.
The term  $\ln |\underline{x}-\underline{x}'|$ contributing to (\ref{eq-13}) can be estimated in the analogous way.
\end{proof}

To recover the second sheet continuation of $
\mathrm{R}_{\alpha, \Sigma \Sigma } (\cdot)$
it remains to construct the analytic extensions of  $\omega_n (z)$ and $\Gamma _n (z)$,~cf.~(\ref{eq-resolalphaembed}).
\begin{lemma}  Given $n\in \N$ the functions
$\omega_n (z)$ and $\Gamma_n (z)$ admit the second sheet continuations $\omega^{\mathit{II}}_n (z)$ and $\Gamma^{\mathit{II}}_n (z)$
to $\Pi_k$ through $J_k= (k^2, (k+1)^2)$, $k\in \N$ defined by
\begin{equation}\label{eq-omega2sheet}
\omega^{\mathit{II}}_n (z; x ):= \frac{1}{2\pi }Z_0 (i\sqrt{z-n^2 }|\underline{x}|)\chi_n (x_3)\,,
\end{equation}
where $Z_0 $ is determined by  (\ref{eq-defZ}), and
\begin{equation} \label{eq-Gamma2sheet}
\Gamma^{\mathit{II}}_n (z)=
     \left\{\begin{array}{lr}
 \frac{1}{2\pi }\left(2\pi \alpha -\psi (1)+\ln \frac{\sqrt{z-n^2}}{2i} \right), & \mathrm{for } \quad n>k \\
 \frac{1}{2\pi }\left(2\pi \alpha -\psi (1)+\ln \frac{\sqrt{z-n^2}}{2i}- \pi i \right), & \mathrm{for }  \quad n\leq k \,.
        \end{array}\right.
\end{equation}
\end{lemma}

\begin{proof} The construction of $\omega_n^{\mathit{II}}(z)$ can be obtained mimicking the arguments from the proof of Lemma~\ref{le-contR}.
\\ To get $\Gamma^{\mathit{II}}(z)$ we first assume $k<n $ and $z=\lambda \pm i\varepsilon$, $\lambda \in (k^2, (k+1)^2)$.   Then
$\ln \frac{\sqrt{\lambda-n^2\pm i 0}}{i}= \ln \sqrt{n^2 -\lambda }$ and, consequently, $\Gamma_n(\lambda +i0)=\Gamma^{\mathit{II}}_n(\lambda -i0)$. \\
Assume now that $n\leq k$. Then we have $\lambda -n^2 >0$ and
$\ln \frac{\sqrt{\lambda -n^2 \pm i0 }}{i}= \ln \sqrt {\lambda -n^2 \pm i0 }  \mp  \frac{\pi}{2} i $ 
which implies
\begin{equation}\label{eq-Gammaasy}
  \Gamma_n(\lambda +i0)=\Gamma^{\mathit{II}}_n(\lambda -i0) =
 \frac{1}{2\pi }\left(2\pi \alpha -\psi (1)+\ln \sqrt{\lambda -n^2}  - \frac{\pi }{2} i \right)\,.
  \end{equation}
This, in view of edge-of-the-wedge
theorem,  completes the proof.
\end{proof}
\\ \\
Henceforth, we assume  that $\epsilon_n \neq k^2$ for any $k,n\in \N$. Suppose $z=\lambda-i\varepsilon$, where $\varepsilon$ is a small non-negative number and $\lambda \in J_k$.
 At most one eigenvalue $\epsilon_l$ can exist in the interval $J_k$.  Assuming that $z \in (\Pi_k \cup J_k )\setminus \epsilon_l $
we define the analytic functions
$ z\mapsto \Gamma_n ^{\mathit{II}}(z)^{-1}$ for
$n\in \N$. Then the second sheet continuation of the resolvent
takes the form
\begin{equation}\label{eq-resolvent2}
   \mathrm{R}^{\mathit{II}} _{\alpha ,\Sigma \Sigma } (z)= \mathrm{R}^{\mathit{II}} _{\Sigma \Sigma } (z)+ \sum_{n=1}^\infty \Gamma_n ^{\mathit{II}}(z)^{-1}(w^{\mathit{II}}_n (\bar{z} ), \cdot )_{\Sigma }
  w^{\mathit{II}}_n (z )\,
\end{equation}
for $z \in (\Pi_k \cup J_k )\setminus \epsilon_l $.
\medskip

\noindent {\bf Notation.} In the following we will avoid the superscript $\mathit{II}$ keeping in mind that
all quantities
depending on $z$  are defined for   second sheet continuation if $\Im z <0$ which admits infinitely many
 branches $\Pi_k$, $k\in \N$.

\medskip

Assume that $\epsilon_l \in J_k$. Having  in mind latter purposes we define
\begin{equation}\label{eq-defA}
A_l (z):= \sum_{n\neq l} \Gamma_n (z)^{-1} (w_n (\bar{z}), \cdot )_{\Sigma_\delta }w_n (z)\,,
\end{equation}
for $z\in (\Pi_k \cup J_k)\setminus \epsilon_l $. The following lemma states the operator
norm asymptotics.
\begin{lemma} \label{le-boundedA}
 Operator $A_l(z)\,:\,L^2 (\Sigma_\delta ) \to L^2 (\Sigma_\delta )$ is bounded and the operator
 norm satisfies
 \begin{equation}\label{eq-boundednormA}
   \|A_l (z)\| \leq C |\Sigma_\delta |\,.
 \end{equation}
 \end{lemma}
 \begin{proof}
 Suppose that $z=\lambda - i\varepsilon $. We derive   the estimates
 \begin{eqnarray}
 \nonumber 
 |(A_l (z)f,f)_{\Sigma_\delta }| &\leq &  \left( \sum_{n\neq l } |\Gamma_{n}(z)^{-1}| \|w_n (z)\|_{\Sigma_\delta }^2 \right) \|f\|_{\Sigma_\delta }^2
 \leq \\ \label{eq-2}&& C  \left( \sum_{n\neq l } \|w_n (z)\|_{\Sigma_\delta }^2 \right) \|f\|_{\Sigma_\delta }^2\,;
 \end{eqnarray}
to obtain (\ref{eq-2})
we use  (\ref{eq-Gamma2sheet}).
Now our aim is to show
\begin{equation}\label{eq-aux3}
   \sum_{n\neq l } \|w_n (z)\|_{\Sigma_\delta }^2  \leq C  |\Sigma_\delta |\,.
\end{equation}
To find a bound for the left hand side of (\ref{eq-aux3}) we analyse first
the behaviour of $w_n (z)$ for large $n$ and $z\in (\Pi_k \cup J_k)\setminus \epsilon_l$.
For this aim we employ (\ref{eq-omega2sheet}) and (\ref{eq-defZ}).
Note that for  $n>k $
function $w_n (z)$ admits the representation:
$$
w_n (z,x)= \frac{1}{2\pi } K_0 (\kappa_n (z)|\underline{x}|) \chi_n (x_3)\,,
$$
where $x\in \Sigma_\delta $.
Using again the large argument expansion  (\ref{eq-K0}) 
and the fact that
$\Re (-i \sqrt{z-n^2})\sim  n$ we get the estimate
$$
|w_n (z,x)| \leq C\e^{-r_{\mathrm{min}}n}\,,
$$
where $r_{\mathrm{min}}= \min _{x\in \Sigma_\delta } |\underline{x}|$. This implies
\begin{equation}\label{eq-aux4}
 \sum_{n >  k\,, n\neq l } \|w_n (z)\|_{\Sigma_\delta }^2  \leq C|\Sigma_\delta |= \mathcal{O} (\delta^2 )\,.
\end{equation}
On the other hand for $n\leq k $ function $w_n (z)$ consists of $K_0 $ and $I_0$, see (\ref{eq-defZ}). Both functions are continuous
on $\Sigma_\delta $ and therefore $\|w_n\|^2_{\Sigma_\delta } \leq C|\Sigma_\delta |$. Since the number of such components is finite, in view of
(\ref{eq-aux4}),
we come to (\ref{eq-aux3}) which completes the proof.
 \end{proof}

 \section{Complex poles of resovent} \label{sec-reson}

Assume that $\epsilon_l \in J_k$. Suppose that $\delta $ is sufficiently small.  It follows from Lemmae~\ref{le-boundedR}~and~\ref{le-boundedA} that the operators
 $I-\beta \mathrm{R}_{\Sigma_\delta \Sigma_\delta } (z)$ and $I-\beta A_l (z)$ acting in $L^2 (\Sigma_\delta )$
 are invertible for $z\in (\Pi_k \cup J_k)\setminus \epsilon_l$ and it makes sense to introduce auxiliary notation
  $$G_{\Sigma_\delta} (z):= (I-\beta  \mathrm{R}_{\Sigma_\delta \Sigma_\delta } (z) )^{-1}\,.$$
  Since the norm of $ G_{\Sigma_\delta }(z) A_l (z)\,:\,L^2 (\Sigma_\delta) \to L^2 (\Sigma_\delta) $ tends to $0 $ if $\delta\to 0 $ therefore
  the operator $I +\beta G_{\Sigma_\delta }(z) A_l (z)$ is invertible as well.
  \\ The following theorem ``transfer'' the analysis of resonances from the operator equation
  to the complex valued function equation.
 \begin{theorem} \label{th-ker}
   Suppose  $\epsilon_l \in J_k $ and assume that $z\in (\Pi_k \cup J_k )\setminus \epsilon_l$.  Then the condition
   \begin{equation}\label{eq-resoncondmain}
 \ker (
   I-\beta \mathrm{R}_{\alpha , \Sigma_\delta \Sigma_\delta } (z))  \neq \{0\}
   \end{equation}  is equivalent to
   \begin{equation}\label{eq-reoncond}
     \Gamma_l (z) +\beta ( w_l (\bar{z}), T_l (z)w_l (z) ) _{\Sigma_\delta }=0\,,
   \end{equation}
where
   $$
T_l (z) :=  (I -\beta G_{\Sigma_\delta }(z) A_l (z))^{-1}G_{\Sigma_\delta }(z)\,.
   $$
    \end{theorem}
   \begin{proof}
   The strategy of the proof is partially based on the idea borrowed from~\cite{Chu}.
The following equivalences
 \begin{eqnarray} \nonumber
    I-   \beta \mathrm{R}_{\alpha ,  \Sigma_\delta \Sigma_\delta }(z) &=&\\ \nonumber
  &=&(I- \beta \mathrm{R}_{ \Sigma_\delta \Sigma_\delta }(z) )
   \left(I -\beta G_{\Sigma_\delta }(z) A_l (z)\right. \\ \nonumber  && \left.   -\beta \Gamma_l (z)^{-1}(w_l (\bar{z} ), \cdot )_{\Sigma_\delta } G_{\Sigma_\delta }(z)w_l (z) \right)\\ \nonumber
   &=& (I- \beta \mathrm{R}_{ \Sigma_\delta \Sigma_\delta }(z) )(I  - \beta G_{\Sigma_\delta }(z) A_l (z)) \times \\  \nonumber && \left[ I  -\beta \Gamma_l (z)^{-1}(w_l (\bar{z} ), \cdot )_{\Sigma_\delta } T_l (z)w_l (z) \right]
   \,,
\end{eqnarray}
 show that (\ref{eq-resoncondmain}) is equivalent to
 $$
\ker  \left[ I  -\beta \Gamma_l (z)^{-1}(w_l (\bar{z} ), \cdot )_{\Sigma_\delta } T_l (z)w_l (z) \right] \neq \{0 \}\,.
 $$
 The above condition is formulated for  a rank one operator and, consequently, it is equivalent
 to (\ref{eq-reoncond}).
   \end{proof}
   \\
   Theorem~\ref{th-ker} shows that the problem of complex poles of resolvent $R_{\alpha, \beta }(z)$ can be shifted
   to  the problem of the roots analysis of
   \begin{equation}\label{eq-resonII}
    \eta_l (z,\delta )=0 \,,\quad \mathrm{where }\quad  \eta_l (z,\delta ):=\Gamma_l (z) - \beta \vartheta_l (z,\delta )\,,
   \end{equation}
   and
   $$
   \vartheta_l (z,\delta ):=   (  w_l (\bar{z}), T_l (z)w_l (z) ) _{\Sigma_\delta }\,.
   $$
   The further discussion is devoted to figuring out roots
   of  (\ref{eq-resonII}). In the following we apply the expansion  $(1+A)^{-1}=(1-A+A^2-A^3...)$ valid if $\|A\|<1$.
   Taking $-\beta \mathrm{R}_{\Sigma_\delta \Sigma_\delta } (z)$ 
   as $A$ we get
 \begin{equation} \label{eq-exp1aa}
   G_{\Sigma_\delta } (z)=(I- \beta \mathrm{R}_{\Sigma_\delta \Sigma_\delta } (z))^{-1}=I+ \breve{\mathrm{R}}(z)\,,\quad
   \breve{\mathrm{R}}(z):=\sum_{n=1}(\beta \mathrm{R}_{\Sigma_\delta \Sigma_\delta } (z))^n \,.
 \end{equation}
Expanding the analogous sum for $-\beta G_{\Sigma_\delta } (z) A_l (z)$ one obtains
 \begin{equation}\label{eq-exp1aaa}
   (I- \beta G_{\Sigma_\delta } (z) A_l (z))^{-1}= I+   \beta A_l (z) + \beta\mathrm{\breve{R}}(z)A_l (z) +...
 \end{equation}
   In view of Lemmae~\ref{le-boundedA} and~\ref{le-boundedR} the norm of $\mathrm{R}_{\Sigma_\delta \Sigma_\delta} (z)A_l (z)$ behaves
   as $o(1) \| A_l (z)\|_{\Sigma_\delta  }$ for $\delta \to 0$ and the same asymptotics holds for the operator norm of
   $\mathrm{\breve{R}}(z)A_l (z)$.
    The further terms in (\ref{eq-exp1aaa}) are of smaller order with respect
   to $\delta $.
    Consequently,  applying again (\ref{eq-exp1aaa}) we conclude that $T_l (z)$ admits the following expansion
   \begin{equation}\label{eq-exp1}
   T_l (z)= I+\beta A_l (z)+ \mathrm{\breve{R}}(z)+...
   \end{equation}
   Using the above statements we can formulate the main result.
   \begin{theorem}
     Suppose that $\epsilon_l \in J_k $ and consider the function $\eta_l (z,\delta )\,:\, \Pi_k \cup J_k \times [0, \delta_0)\to \C$,
     where $\delta_0>0$,  defined by  (\ref{eq-resonII}). Then the equation
     \begin{equation}\label{eq-sought}
     \eta_l (z, \delta ) =0\,,
     \end{equation}
     possesses a solution which is determined by the function  $ \delta \mapsto z(\delta )\in \C$ with the following asymptotics
   \begin{equation}\label{eq-1}
     z_l(\delta )= \epsilon_l+\mu_l (\delta )\,, \quad |\mu_l(\delta )|= o(1)\,.
   \end{equation}
Moreover, the lowest order term of $\mu_l (\cdot )$ takes the form
\begin{eqnarray} \label{eq-asympmu}
\mu_l(\delta ) = && 4\pi \xi_\alpha \beta  \left\{
\|w_l (\epsilon_l )\|^2_{\Sigma_{\delta }} \right.
\\ \label{eq-asym1} && \left.
+\beta
\sum_{n\neq l }  \Gamma_n (\epsilon_l )^{-1} |(w_l (\epsilon_l), w_n (\epsilon_l))_{\Sigma_\delta }|^2
\right.
\\ \label{eq-asym2} &&
\left.  +   (w_l (\epsilon_l ),\mathrm{\breve{R}}(\epsilon_l) w_l (\epsilon_l ) )_{\Sigma_\delta }   \right\}\,.
\end{eqnarray}
   \end{theorem}
   \begin{proof}
   Note that  $z\mapsto \eta_l (z,\delta )$, cf.~(\ref{eq-resonII}), is analytic and
   $\eta_l (\epsilon_l , 0) = 0 $. Using  (\ref{eq-Gamma2sheet}) 
   one obtains $$\left. \frac{d \Gamma _n (z)}{dz}\right|_{z=\epsilon_n }=\frac{1}{4\pi  \xi_\alpha } <0 \,, \quad n\in \N.$$
Combining this with
 $$\left. \frac{\partial  \vartheta_{n}(z,\delta )}{\partial z} \right|_{z=\epsilon_l, \delta =0 } = 0\,, $$
   we get
     $\left.\frac{\partial \eta_l }{\partial z}\right|_{\delta=0}= \frac{1}{4\pi \xi_\alpha }\neq 0$.
     In view of the Implicit Function Theorem we conclude that the equation (\ref{eq-resonII})
      admits a unique solution
     which a continuous function of $\delta \mapsto z_l(\delta)$ and $z_l (\delta )= \epsilon_l +o(1)$.
     To reconstruct asymptotics of $z(\cdot )$
     first we expand $\Gamma_l (z)$  into the Taylor sum
     $$
     \Gamma_l (z)= \frac{1}{4\pi \xi_\alpha } (z-\epsilon_l) +\mathcal{O}((z-\epsilon_l)^2)\,. 
     $$
     Then the spectral equation (\ref{eq-resonII}) reads
     $$
     z=\epsilon_l + 4 \pi\xi_\alpha  \beta \vartheta_l (z,\delta )+\mathcal{O}((z-\epsilon_l)^2)\,.
     $$
      Now we expand $\vartheta_l (z,\delta )$. Using  (\ref{eq-exp1}) and (\ref{eq-defA})
   we reconstruct its  first order term which reads
      \begin{eqnarray} \nonumber 
   && \left\{
\|w_l (\epsilon_l )\|^2_{\Sigma_{\delta }}
+\beta
\sum_{n\neq l }  \Gamma_n (\epsilon_l )^{-1} |(w_l (\epsilon_l), w_n (\epsilon_l))_{\Sigma_\delta }|^2
\right.
\\ \nonumber  && \left.
   + (w_l (\epsilon_l ),\mathrm{\breve{R}}(\epsilon_l) w_l (\epsilon_l ) )_{\Sigma_\delta }  \right\} \,.
\end{eqnarray}
       Applying the asymptotics
     $z_l(\epsilon_l)=\epsilon_l +o(\delta )$ and the fact that $\vartheta_l (\cdot ,\cdot  )$ is analytic with respect to
     complex variable we get formula for $\mu(\cdot)$.
   \end{proof}

   \subsection{Analysis of imaginary part of the pole }

Since the imaginary component of resonance pole has a physical meaning
we dedicate to this problem a special discussion.
The information on the lowest order term of the pole imaginary component is contained in    (\ref{eq-asym1}) and (\ref{eq-asym2}). On the other  hand, note that
only the components subscripted by $n\leq k$ admit a non-zero imaginary parts. Therefore
\begin{eqnarray} \label{eq-imaginary}
\Im && \left( 4\pi \xi_\alpha \beta  \left( \beta
\sum_{n\leq k  } \Gamma_n (\epsilon_l )^{-1} |(w_l (\epsilon_l), w_n (\epsilon_l))_{\Sigma_\delta }|^2
\right. \right.\\ \label{eq-imaginary1} && \left. \left.+ (w_l (\epsilon_l ),\mathrm{\breve{R}}(\epsilon_l) w_l (\epsilon_l ) )_{\Sigma_\delta } \right) \right) \,.
\end{eqnarray}
determines the lowest order term of $\Im \mu (\delta )$.
\\ \\
\emph{Sign and asymptotics of $\Im \mu (\delta )$  with respect to $\Sigma_\delta$ }.
Recall that $\epsilon_l \in J_k$. First we analyse (\ref{eq-imaginary}) and for this aim we define
$$
\iota_{l,n}: = \frac{1}{2\pi }\left( 2\pi \alpha +\ln \frac{\sqrt{\epsilon_l -n^2}}{2}-\psi(1)\right)\,,
$$
for $n\leq k$.
Relying on (\ref{eq-Gammaasy}) we get
$$
\Gamma_l (\epsilon_l )^{-1} = \frac{1}{\iota_{l,n}^2 +(1/2)^2}\left(\iota_{l,n} +\frac{1}{2}i \right)
$$
if  $n\leq k $. Consequently,  formula (\ref{eq-imaginary}) is equivalent to
$$
\Im \, 4\pi \xi_\alpha \beta^2
\sum_{n \leq k }  \frac{1}{2}\frac{1}{\iota_{l,n}^2 +(1/2)^2}
  |(w_l (\epsilon_l), w_n (\epsilon_l))_{\Sigma_\delta }|^2\,.
$$
The above expression is negative because  $\xi_\alpha <0$.
Moreover, since both $w_l (\epsilon_l)$ and $w_l (\epsilon_n)$ are continuous in $\Omega\setminus I$
we have $|(w_l (\epsilon_l), w_n (\epsilon_l))_{\Sigma_\delta }|^2\sim |\Sigma_\delta |^2 $. This means that  (\ref{eq-imaginary}) behaves  as $\mathcal{O}(|\Sigma_\delta|^2)$. To recover the asymptotics of (\ref{eq-imaginary1}) we
restrict ourselves  to the lowest order term of $\breve{\mathrm{R}}(z)$, cf.~(\ref{eq-exp1}), namely
$$
 \upsilon_l := \Im 4\pi \xi_\alpha   \beta ^2 (w_l (\epsilon_l ),\mathrm{R}_{\Sigma_\delta \Sigma _\delta }(\epsilon_l) w_l (\epsilon_l ) )_{\Sigma_\delta }\,.
$$
Using analytic continuation formulae   (\ref{eq-Kext}) and employing the small argument expansion, cf.~\cite{AS},
$$
K_0 (z) \sim  -\ln z\,,
$$
where $-\pi <\arg z <\pi$ states the plane cut for the logarithmic
 function,
one gets
$$
 \upsilon_l \sim \Im  \pi \xi_\alpha   \beta ^2
\sum_{n\leq  k } \left( \int_{\Sigma_\delta } w_l (\epsilon_l )\chi_n \right)^2 =\mathcal{O}(|\Sigma_\delta |^2 )\,.
$$

One can easily see that $ \upsilon_l <0$.
Summing up the above discussion we can formulate the following conclusion.

\begin{proposition}
The resonance pole takes the form $z_l (\delta )=\epsilon_l+\mu (\delta )$ with the lowest order of $\Im \mu (\delta )$ given by
$$
\pi \xi_\alpha   \beta ^2   \sum_{n \leq k }  \left(  \frac{2}{\iota_{l,n}^2 +(1/2)^2}
  |(w_l (\epsilon_l), w_n (\epsilon_l))_{\Sigma_\delta }|^2 +
 \left( \int_{\Sigma_\delta } w_l (\epsilon_l )\chi_n \right)^2 \right) \,.
$$
It follows from the above formula that
$\Im \mu (\delta )<0$ and
the asymptotics $$ \Im  \mu (\delta )= \mathcal{O}(|\Sigma_\delta |^2) \,$$ holds. Moreover, the lowest order
of $\Im \mu (\delta )$ is independent of sign of $\beta $.
\end{proposition}


Note that for the special geometrical cases the embedded eigenvalues can survive
after introducing $\Sigma_\delta $ since the "perturbed" eigenfunctions are not affected
by presence of $\Sigma_\delta $.
Let us consider
 $$\Pi_l := \{x\in \Omega \,:\, x=\left(\underline{x}, \frac{\pi }{l }\right)\,, \quad l\in \N$$
 and assume that $\Sigma_\delta  \subset \Pi_l $. Then $w_{ml} (z) = 0 $ for each $m\in \N$ and, consequently
 $\vartheta_{ml} (z,\delta ) = 0$, cf.~(\ref{eq-resonII}). This implies the following statement.

\begin{proposition}
  Suppose that $\Sigma \subset \Pi_l $. Then
  for all $m\in \N$ the numbers $\epsilon_{ml}$ remain the embedded eigenvalues of $H_{\alpha , \beta }$.
\end{proposition}

\subsection*{Acknowledgements}
The  author  thanks  the  referees  for  reading  the  paper  carefully, removing errors and
recommending
various improvements in exposition. \\
The work was  supported
 by the project DEC-2013/11/B/ST1/03067 of the Polish National Science Centre.

\end{document}